\documentclass[12pt, english]{amsart}
\usepackage[inner=1in,outer=1in,top=1in,bottom=1in, marginparwidth=.7in,marginparsep=.1in]{geometry}
\usepackage{algpseudocode}
\usepackage{amsmath}
\usepackage{verbatim}
\usepackage{mathrsfs}
\usepackage{mathtools}
\usepackage{aliascnt}
\usepackage[mathcal]{eucal}
\usepackage{microtype}
\usepackage{graphicx,multirow}
\usepackage[all]{xy}
\usepackage{tikz}
\usepackage[utf8x]{inputenc} 
\usepackage{enumerate,caption}
\usepackage{enumitem}
\usepackage{xspace}
\usepackage[modulo]{lineno}
\usepackage{subcaption}

\usepackage[linktocpage,colorlinks=true,linkcolor=blue,citecolor=blue,urlcolor=blue,citebordercolor={0
  0 1},urlbordercolor={0 0 1},linkbordercolor={0 0 1}]{hyperref} 
\usepackage[nameinlink]{cleveref}
\crefname{diagram}{Diagram}{Diagram}
\usepackage[normalem]{ulem}
\usepackage{soul}

\usepackage{color}

\numberwithin{equation}{section}
\usepackage{amssymb,latexsym,amsmath,amsthm,mathrsfs}
\usepackage{tikz-cd,color}
\theoremstyle{plain}
\newtheorem{theorem}{Theorem}[section]
\newtheorem{corollary}[theorem]{Corollary}
\newtheorem{lemma}[theorem]{Lemma}

\newtheorem{proposition}[theorem]{Proposition}

\newtheorem*{thm*}{Theorem}

\theoremstyle{definition}

\newtheorem{remark}[theorem]{Remark}

\newtheorem{defn}[theorem]{Definition}

\newtheorem{example}[theorem]{Example}

\newcommand\I{\iota}

\newcommand\C{\mathbb{C}}

\newcommand\N{\mathbb{N}}

\newcommand\bV{\mathbb{V}}
\newcommand\bW{\mathbb{W}}

\newcommand\R{\mathbb{R}}
\newcommand\Z{\mathbb{Z}}

\DeclareMathOperator\Sym{Sym}

\DeclareMathOperator\Id{Id}
\DeclareMathOperator\Tr{Tr}
\DeclareMathOperator\GL{GL}
\DeclareMathOperator\sym{sym}

\DeclareMathOperator\diag{diag}

\DeclareMathOperator\End{End}
\DeclareMathOperator\band{band}

\DeclareMathOperator{\Hom}{Hom}

\DeclareMathOperator{\SL}{SL}

\DeclareMathOperator{\Sp}{Sp}
\DeclareMathOperator{\SO}{SO}
\DeclareMathOperator{\SU}{SU}
\DeclareMathOperator{\Spin}{Spin}

\title{Orbit recovery for band-limited functions}
\author{Dan Edidin, Matthew Satriano}
\address{Department of Mathematics, University of Missouri, Columbia MO 65211}
\email{edidind@missouri.edu}
\address{University of Waterloo, Department of Pure Mathematics, Waterloo, Ontario, Canada N2L 3G1}
\email{msatrian@uwaterloo.ca}
\begin{document}
\begin{abstract}
  We study the third moment for functions on arbitrary compact Lie
  groups.  We use techniques of representation theory to generalize
  the notion of band-limited functions in classical Fourier theory to
  functions on the compact groups $\SU(n), \SO(n), \Sp(n)$. We
  then prove that for
  generic band-limited functions the third moment or, its Fourier
  equivalent, the bispectrum determines the function up to translation by a
  single unitary matrix.
  Moreover, 
  if $G=\SU(n)$ or $G=\SO(2n+1)$ we prove that the third moment determines
  the $G$-orbit of a band-limited function. As a corollary we obtain a
  large class of finite-dimensional representations of these groups for which the third
  moment determines the orbit of a generic vector. When $G=\SO(3)$ this
  gives a result relevant to cryo-EM which was our original motivation
  for studying this problem.
\end{abstract}
% REQUIRED
%\begin{keywords}
%Orbit recovery, bispectrum
%\end{keywords}

% REQUIRED
%\begin{MSCcodes}
%94A12, 22D10
%\end{MSCcodes}
\maketitle
\section{Introduction}
Let $G$ be a compact Lie group. The purpose of this paper is to construct a class
of finite-dimensional representations $V$ of $G$ for which the third moment
can determine the orbit of a generic vector $f \in V$.
As we explain, this 
work is motivated by several
applications including multi-reference alignment (MRA), cryo-EM, and machine learning.

In its basic form, the multi-reference alignment (MRA) problem
seeks to recover a signal $f \in V$ from noisy group translates
of the signal
\[ y_i = g_i\cdot f + \epsilon_i\]
where the $g_i$ are randomly selected from a uniform distribution on $G$
and the $\epsilon_i$ are taken from a Gaussian distribution ${\mathcal N}(0, \sigma^2I)$
which is independent of the group element $g_i$.
Without prior knowledge of the group elements, there is no way
to distinguish $f$ from $g\cdot f$ for any $g \in G$. Thus
the MRA problem is one of orbit recovery. The MRA problem has been extensively studied in recent years, beginning with action of $\Z_N$ on $\R^N$ by cyclic
shifts \cite{bandeira2014multireference, perry2019sample, bendory2017bispectrum,abbe2018multireference, bandeira2020optimal}. Other models include the dihedral group \cite{bendory2022dihedral}
and the rotation group $\SO(2)$ acting on band-limited functions on $\R^2$
\cite{bandeira2020non,ma2019heterogeneous,janco2022accelerated}. The case where $G = \SO(3)$ is particularly important
because of its connection to cryo-EM, a leading technique in molecular imaging.
There, the measured data can be modelled as $y_i = T (g \cdot f) + \epsilon$
where $f$ is the Coulomb potential of an unknown molecular structure and
$T$ is a tomographic projection \cite{bendory2020single}.

In the low-noise regime the products $g_ig_j^{-1}$ can be estimated
using the method of synchronization and then the signal can be
approximated by averaging \cite{singer2011angular}. However, in the
high-noise regime, as is the case for cryo-EM measurements of small
molecules, there is no way to accurately estimate the relation between
the unknown group elements \cite[Proposition
  2.1]{bendory2018toward}. One common approach to this problem is to
use the {\em method of moments}.  In this case, it can be shown
\cite{perry2019sample} that the moments of the unknown signal can be
accurately approximated by the computing the corresponding moments of
the experimental data. Thus, a crucial aspect of the MRA problem is
understanding how to recover a signal $f \in V$ from its moments which
are, by definition, $G$-invariant tensors in the signal.

In machine learning it is desirable to build neural networks whose
archicture reflects the instrinsic structure of the data. When the
data has natural symmetries under a group $G$ then we want to build
the network from {\em $G$-equivariant functions}. (Recall that if
$\bV$ and $\bW$ are sets with a $G$ action then a function $f \colon
\bV \to \bW$ is $G$-equivariant if $f(gv) = gf(v)$ for all $g \in G$.)
The basic model of an equivariant neural network
\cite{cohen2016group,lim2022equivariant} is a sequence of maps
\[\R^{n_0} \stackrel{A_1} \to \R^{n_1} \stackrel{\sigma_{b_1}} \to
\R^{n_2} \ldots \stackrel{\sigma_{b_{k-1}}} \to \R^{n_{k-1}}
\stackrel{A_k} \to \R^{n_k}\] where each $\R^{n_i}$ is a
representation of $G$, the $A_i$ are $G$-equivariant linear
transformations and the $\sigma_{b_i}$ are non-linear maps.

 One difficulty with this model is that there may be relatively
 few $G$-equivariant linear maps of representations $A_i \colon \R^{n} \to \R^{m}$, so an equivariant neural network built this way may not be sufficiently expressive. Two common, and mathematically related, ideas
 are to use equivariant
 linear maps on tensors $(\R^{n})^{\otimes k} \to
 (\R^{m})^{\otimes \ell}$, or invariant 
 polynomials of degree $\ell$ from $\R^n \to \R^m$ \cite{kondor2018clebschgordan, dym2020universality,blumsmith2023machine}.
 
In both MRA and machine learning the use of tensors is theoretically
 desirable, but the cost of computing tensors grows
 exponentially in the degree. In addition, in MRA
 the sample complexity (the minimum number of observations required for accurate approximation) grows
 as $\sigma^{2d}$ where $d$ is the number of moments used and $\sigma^2$
 is the variance of the noise. For this reason, we wish
 to identify representations for which moments of low degree separate
 generic orbits. Previous work of Bendory and the first author demonstrated that
 neither the first or second moment carries enough information to
 separate orbits in all but the simplest representations
 \cite{bendory2022sample}. Thus an important problem is to understand
 and classify representations of compact groups for which the third
 moment can separate orbits. Previous work
 \cite{bandeira2017estimation} showed that for any finite group the
 generic orbit in the regular representation can be recovered from the
 third moment using Jennrich's algorithm. More generally it is proved
 in \cite{smach2008generalized} that if $G$ is a positive dimensional compact group then 
 the $G$-orbit of a function with non-singular Fourier coefficients in the infinite dimensional representation
 $L^2(G)$ can recovered from its third moment.

 In this paper we show that if $G$ is one of the classical groups
 $\SU(n), \SO(n), \Sp(n)$ then it is possible to use representation theory
 to define the notion
 of {\em band-limited function} which generalizes the notion
 of band-limited function in Fourier theory as well as
 previous defintions for $G=\SO(3)$. 
 Our main result can be stated as follows.
 \begin{theorem}[Informal] \label{thm.informal}
   (i) If $G$ is one of the classical groups then the generic band-limited
   function can be recovered up to translation by a single unitary matrix from its third moment.
   
   (ii) The orbit of any generic
   band-limited function in $L^2(\SU(n))$ or real-valued band-limited
   function in $L^2(\SO(2n+1))$ can be recovered from its third moment.
 \end{theorem}
\begin{remark}
As is the case in Fourier theory, for
any band level, the vector space of band limited functions is
finite dimensional. When we say that the orbits of generic band-limited functions
can recovered from their third moments, we mean that the set of orbits which
cannot be recovered is contained in a proper real algebraic
subset of this finite-dimensional vector space.
A precise sufficient condition in terms of the non-singularity of
Fourier coefficients is given in Theorem~\ref{thm.suso}.

We note that our result gives a large class of finite-dimensional representations for which the third moment separates generic orbits.
 \end{remark}

The proof of Theorem~\ref{thm.informal} uses representation theory to
generalize the well-known frequency marching result that states that a
band-limited function on $S^1$ can be recovered from its bispectrum.
A short discussion of potential algorithms using generalized frequency
marching is given in Section~\ref{sec.algorithms}.
When $G= \SO(3)$ Kakarala \cite{kakarala2012bispectrum} showed that
the $\SO(3)$-orbit of a band-limited real valued function can be
recovered up to reflection from the bispectrum, and our result
eliminates the reflection ambiguity. Indeed the idea of using
techniques of representation theory to study the orbit recovery
problem was motivated by Kakarala's earlier work for $\SO(3)$.

In Section~\ref{sec.so3} we
  focus on the group $\SO(3)$, since results for this group 
  have the most potential
  applications. In particular, we 
  compare our work with recent work of Liu and Moitra~\cite{liu2021algorithms} for recovering band-limited functions in $L^2(S^2)$.
We also prove (Corollary \ref{cor.cryo})
that if we consider the finite dimensional approximation
of $L^2(\R^3)$ as band-limited
functions on $R$ spherical shells, where the number of shells exceeds
the band-limit then the $\SO(3)$ orbit of a generic function can be recovered
from the third moment. This is a standard assumption in the cryo-EM literature
\cite{bandeira2020non} and our result affirms many cases of a conjecture
made in \cite{bandeira2017estimation}.

 \section{Moments of representations of compact groups}
Let $V$ be a unitary representation of a compact group $G$.
The $d$-th moment of $V$ is the function
$m_d \colon V \to \underbrace{V \otimes \ldots \otimes V}_{\text{ $d -1$ times}} \otimes V^*$
defined by the formula
\begin{equation} \label{eq.moment}
  f \mapsto m_d(f) \colon = \int_G \overbrace{g\cdot f \otimes \ldots \otimes
  g\cdot f}^{\text{$d -1$ times}} \otimes \overline{g\cdot f}\;dg
\end{equation}
Note that the formula of \eqref{eq.moment} is invariant under translation
by $G$ so for any $f \in V$, $m_d(f)$ a $G$-invariant element
of the tensor $V^{d-1} \otimes V^*$ or equivalently a $G$-invariant element of
$\Hom(V, V^{d})$.

If $V$ is a real representation then
the $d$-th moment is an element of $\Sym^d V$ and the collection
of moments form a set of generators for the invariant ring $\R[V]^G$.

When $V$ is identified with a space of functions $D \to \C$ where
$D$ is a domain on which $G$ acts (for example we can take $D=G$ and
consider $V$ as a subspace of $L^2(G)$)
then $m_d(f)$ is viewed as a function $D^d \to \C$ defined by
the formula
\begin{equation} \label{eq.momentfunction}
  m_d(f)(x_1, \ldots , x_d) = \int_G (g\cdot f)(x_1) \ldots (g\cdot f)(x_{d-1})
  \overline{(g\cdot f)(x_d)}\;dg
\end{equation}
where $g\cdot f \colon D \to \C$ is the function 
$(g\cdot f)(x) = f(g^{-1}x)$.

\subsection{Moments and the decomposition into irreducibles}
A general finite dimensional representation of a compact group
can be decomposed as $V = \oplus_{\ell =1}^L V_\ell^{R_\ell}$ where
the $V_\ell$ are distinct non-isomorphic irreducible representations
of $G$ of dimension $N_\ell$. An element $f \in V$ has a unique
$G$-invariant decomposition as a sum
\begin{equation}
  f = \sum_{\ell=1}^L \sum_{i=1}^{R_\ell} f_\ell[i]
\end{equation}
where $f_\ell[i]$ is in the $i$-th copy of the irreducible
representation $V_\ell$.
For fixed $\ell$ the vectors $f_\ell[1], \ldots , f_\ell[R_\ell]$
determine a $N_\ell \times R_\ell$ matrix $A_\ell(f)$ which we sometimes refer
to as the {\em coefficient matrix of $f$ in $V_\ell$}.

The $d$-th moment is a map $m_d \colon V \to
(V^{\otimes (d-1)} \otimes V^*)^G$, which takes a vector $f \in V$ to
invariant part of the tensor $f \otimes f \otimes \ldots \otimes
\overline{f}$. Identifying $V^{\otimes (d-1)} \otimes V^* =
\Hom(V,V^{\otimes d-1})$, the $d$-th moment decomposes according to the
  decomposition of $V$ into irreducibles as a functions
  $$m_d[V_{i_1}, \ldots , V_{i_d}] \colon V \to \oplus_{i_1,\ldots , i_d} \Hom(W_{i_d}, W_
  {i_1} \otimes \ldots \otimes  W_{i_{d-1}}))^G$$
  where $W_\ell = V_\ell^{R_\ell}$.
By Schur's lemma $m_d[V_{i_1}, \ldots , V_{i_d}]$ is non-zero if and only if the irreducible representation $V_{i_d}$ is 
a summand in the tensor product $W_{i_1} \otimes \ldots \otimes W_{i_{d-1}}$.

  The map $m_d[V_{i_1}, \ldots , V_{i_d}]$ can be described explicitly
  as follows: If $f =\sum_{\ell =1}^L A_\ell(f)$ where $A_\ell(f) =
  (f_\ell[1], \ldots , f_\ell[R_L]) \in W_\ell$, let $B_{i_d}(f)$ be
  the projection of \[A_{i_1}(f) \otimes \ldots \otimes A_{i_{d-1}}(f)
  \in W_{i_1} \otimes \ldots W_{i_{d-1}}\] to the isotypic component,
  $T_{i_d}$ of $W_{i_1} \otimes \ldots \otimes W_{i_{d-1}}$
  corresponding to the irreducible representation $V_{i_d}$.  Then
  $m_d[V_{i_1}, \ldots , V_{i_d}](f)$ is the element of $\Hom(W_{i_d},
  T_{i_d})$ represented by the matrix $A_{i_d}(f) B_{i_d}(f)^*$.
  
In the case when $d=2$ there is a simple description of the information determined by the second moment.
\begin{proposition} \label{prop.secondmoment} \cite[Theorem 2.3]{bendory2022sample}
  Let $V$ be a finite dimensional representation of $G$.
The second moment $m_2(f)$ determines, for each irreducible
$V_\ell$ appearing in $V$, the product $A_\ell(f) A_\ell(f)^*$. In particular, %if $A_\ell(f)$ is full rank,
the second moment determines
  $A_\ell(f)$ up to translation by an element of $U(V_\ell)$, or if $A_\ell(f)$
  is real, an element of $O(V_\ell)$.
\end{proposition}

The next result proves that if $V$ contains a copy of the trivial representation then the third moment determines the first and second moments.  
\begin{proposition} \label{prop.3to2}
If the coefficient matrix, $A_0(f)$, of the trivial representation
is non-zero then 
$m_{1}(f)$ and $m_2(f)$ are determined from $m_3(f)$.
\end{proposition}
\begin{proof}
  Since the trivial representation $V_0$ is one-dimensional the coefficient matrix
  $A_0(f)=(f_0[1], \ldots , f_0[R_0]) $ is just a row vector of length $R_0$ where $R_0$ is the multiplicity. The first moment is simply the projection
  of $V \to V^G = V_0^{R_0}$ so $m_1(f) = A_0(f)$.
  On the other hand, consider
  the component $m_3[V_0,V_0,V_0]$ of the third moment which is a $G$-invariant map
  $V \to W_0 \otimes W_0 \otimes W_0^*$. Since $G$ acts trivially
  on $W_0$, the entire tensor product $W_0 \otimes W_0$ is the $V_0$ isotypic component. Thus
  $m_3(f)[V_0,V_0,V_0] = A_0(f) (A_0(f) \otimes A_0(f))^*$ where we view
  $A_0(f) \otimes A_0(f)$ as a matrix of size $1 \times R_0^2$.
  Among the entries of this matrix are the products $f_0[j]
  \overline{f_0[j]f_0[j]} = \overline{f_0[j]}|f_0[j]|^2$.
    Writing $f_0[j] = r_je^{i\theta}$ for $r_j>0$ we see that $f_0[j]$ is determined
    by $\overline{f_0[j]}|f_0[j]|^2$.

    In particular if $f_0[i] \neq 0$ is known, $m_2(f)$ is determined
    by ${1\over{|A_0(f)|}}\int_G A_0[f] \otimes g\cdot f \otimes \overline
    {g\cdot f}$ which
    is a sum of components of $m_3(f)$. 
\end{proof}

The following bootstrap result is a generalization of
\cite[Proposition 4.15]{bandeira2017estimation}.
  
\begin{proposition} \label{prop.bootstrap}
  Let $V$ be a representation of $G$ with
$V=V_1^{R_1} \oplus \ldots
  \oplus V_L^{R_L}$ 
where $V_1, \ldots , V_L$ distinct irreducibles.
Assume that for every $1 \leq k \leq L$, the generic vector $f$ in
the representation 
$V_1^{R_1}  \oplus \widehat{V_k^{R_k}} \ldots \oplus V_L^{R_L}$
has trivial stabilizer. (Here the notation $\widehat{V_k^{R_k}}$ means
that the summand is omitted.)

If 
the orbit of a generic vector $f \in V$ is determined from
  the $d$-th moment $m_d(f)$ and $W \supset V$ is another representation
  with the same irreducible components then the orbit of a generic vector
  $h \in W$ is determined by its the $d$-th moment $m_d(h)$.
\end{proposition}
\begin{example} \label{ex.fourierband}
  The hypothesis that $W$ has the same irreducible components as $V$
  is necessary. Consider the case where $G=S^1$ and $V =V_0 \oplus V_1$,
  $W = V_0 \oplus V_1 \oplus V_3$
  where the $V_n$ is the one-dimensional representation of $S^1$ where
  $S^1$ acts with weight $n$; i.e. $e^{i \theta } \cdot v = e^{n i \theta}v$.
  If $v = (v_0,v_1) \in V$ then the third moment determines
  $v_0$ and $|v_1|^2$ which determines the vector $v$ up to multiplication
  by an element $S^1$. On the other hand, the third moment of
  $w= (w_0, w_1, w_3) \in  W$ determines, $w_0, |w_1|^2, |w_3|^2$ which is not sufficient to determine the vector $w$ up to multiplication by an element of
  $S^1$. 
\end{example}

\begin{proof}
  By induction on the multiplicities we may reduce to the case that
\[ W=V_1^{R_1} \oplus \ldots V_\ell^{R_{\ell}+1} \oplus \ldots \oplus V_L^{R_L}\;\] i.e.
    all multiplicities of irreducibles in $W$ are the same as in $V$ except
    for the multiplicity of $V_\ell$ which is $R_\ell$ in $V$ and $R_\ell +1$
    in $W$. By reordering the irreducibles we assume that $\ell =1$

    Suppose $h \in W$ has coefficient matrices, $B_1, B_2, \ldots , B_L$
    where $B_1 = (f_1[1], \ldots , f_1[R_1+1])$. For each $j =1, \ldots , R_1+1$
    consider the $G$-invariant projection $\pi_j \colon W \to V$ which sends
    $h$ to the vector $f_j$ with coefficient matrices
    $B_1^j, B_2, \ldots, B_L$ where
    $B_1^j = \left(f_1[1], \ldots, \widehat{f_1[j]}, \ldots , f_1[R_1+1]\right)$
By assumption, the $G$-orbit of $f_j$ is determined from $m_2(f_j) = m_2(\pi_j h))$. 
In particular, if $h'$ is another vector in $W$ with $m_d(h) = m_d(h')$
and coefficient matrices, $B_1', \ldots, B_L'$ 
then there exists $g_1, \ldots , g_{R_1+1} \in G$ such that
$g_j(B_1^j, \ldots, B_{L}) = (B'^{j}_1, \ldots , B'_L)$. To show
that $g_j$'s are all equal note that for any $j_1, j_2$
the vector in the representation
$V_2^{R_2} \oplus \ldots V_L^{R_L}$ with coefficient
matrices $B_2, \ldots , B_{L}$ is fixed by
$g_{j_1}g_{j_2}^{-1}$ so by assumption on the representation we can conclude that $g_{j_1} = g_{j_2}$. Thus, there exists $g \in G$ such that for every
$j$, $g(B_1^j, B_2, \ldots , B_L) = (B'^{j}_1, B'_2, \ldots , B_L)$.
Since every row vector in $B_1$ (resp $B_1'$) is in some matrix
$B_1^j$ (resp. $B_1^{'j}$) it follows
that $g(B_1, \ldots , B_L) = (B_1', \ldots , B'_L)$. In other words,
$h' = gh$ for some $g \in G$.
\end{proof}

\subsection{The Fourier transform on compact groups}
If $G$ is a compact group and $f \in L^2(G)$ then the {\em Fourier
transform} (see \cite{smach2008generalized} as well as \cite[Chapter 8]{hewitt1970abstract}) of $f$ is a matrix-valued function $F(f)$ defined on the
set representations of $G$ by the formula \begin{equation}
\label{eq.fourier} F(f)(V) = \int_G f(g)D_V(g)^*\;dg \in \End(V)
\end{equation} where $D_V(g)$ the unitary linear transformation $V
\to V$ defined by $v \mapsto gv$.

The matrix $F(f)(V)$ is called the {\em Fourier coefficient} of $V$.
Later, we will implicitly choose a basis for each irreducible representation
so we can view the Fourier coefficient as a matrix.  As is the case for the classical Fourier transform, a
function $f \in L^2(G)$ is uniquely determined by its
Fourier coefficients $F(f)(V)$ where $V$ runs through all irreducible
representations of $G$ \cite[Theorem 31.5]{hewitt1970abstract}.

Conversely, if $V$ is a representation of $G$ and $T \in \End(V)$
is an endomorphism, then the inverse Fourier transform of $T$ is
the function
\begin{equation} \label{eq.fourierinv} f_T(g) = {1\over{\dim V}} \Tr(T D_v(g)^*).
\end{equation}.

\subsection{The regular representation and higher-order spectra}
Here we take $V = L^2(G)$ to be the regular representation.
In this case
$$m_d(f)(g_1, \ldots , g_d) = \int_G f(g^{-1}g_1)f(g^{-1}g_2) \ldots
f(g^{-1}g_{d-1})
\overline{f(g^{-1}g_d)}\;dg$$
  Applying the change of variables $g = g^{-1}g_d$ 
  we can rewrite
 \[m_d(f)(g_1, \ldots , g_d) = m_d(f)(g_d^{-1}g_1, \ldots ,
  g_d^{-1}g_{d-1},1).\]
  Hence after replacing $g$ with $g^{-1}$ we may view the $d$-th moment of the regular
  representation as the function on $G^{d-1}$.
  \begin{equation} \label{eq.moment-regular}
   m_d(f)(g_1, \ldots , g_{d-1})  = \int_G f^*(g) f(gg_1)
   \ldots  f(gg_{d-1}) \;dg
  \end{equation}
where $f^*(g)=\overline{f(g)}$.

The {\em $(d-1)$-st higher-spectrum}
$a_d(f)$ is defined
as the Fourier transform of the function $m_d(f)
\in L^2(G^{d-1})$.
Since every irreducible representation of $G^{d-1}$ is of the form
$V_1 \otimes \ldots \otimes  V_{d-1}$ we have
that
\begin{multline} \label{eq.highspectrum}
 a_d(f)(V_1 \otimes \ldots \otimes  V_{d-1}) 
  =\int_{G^{d-1}} \left(\int_G f^*(g) f(gg_1) \ldots
f(gg_{d-1})\;dg\right) \times \\D_{V_1}(g_1)^* \otimes \ldots D_{V_{d-1}}(g_{d-1}^*)\;dg_1\ldots dg_{d-1}
\end{multline}
Using the change of coordinates where replace $g_i$ with $gg_i$
and reversing the order of integration the right-hand side of \eqref{eq.highspectrum} becomes
\begin{multline}
\int_G  \int_{G^{d-1}} f(g_1) \ldots f(g_{d-1})\times \\ [D_{V_1}(g_1)^* \otimes
\ldots \otimes
D_{V_{d-1}}(g_{d-1})^*)]\;dg_1 \ldots d_{g_{d-1}}
[D_{V_1}(g) \otimes \ldots \otimes D_{V_{d-1}}(g)] f^*(g) \;dg  \\
 =  [F(f)(V_1) \otimes \ldots \otimes F(f)(V_{d-1})]
\int_G f^*(g)[D_{V_1}(g) \otimes \ldots \otimes D_{V_{d-1}}(g)]\;dg\\
 = [F(f)(V_1) \otimes \ldots \otimes F(f)(V_{d-1})]
 [F(f)(V_1 \otimes \ldots V_{d-1})]^*
\end{multline}
where the product is taken in the ring $\End(V_1 \otimes \ldots \otimes V_{d-1})$.
Since the $(d-1)$-st higher spectrum is the Fourier transform of the $d$-th
moment and the Fourier transform is invertible the $(d-1)$-st
higher spectrum carries the same information as the $d$-th moment.
\subsection{The bispectrum}
When $d=3$ the third moment $m_d(f)$ carries the same information
as the bispectrum
$a_2(f)$ whose value on a tensor $V \otimes W$
is
\begin{equation} \label{eq.bispectrum}
  a_2(f)(V \otimes W)  =  [F(f)(V) \otimes F(f)(W)] [F(f)(V \otimes W)]^*
%  & = & [F(f)(V) \otimes F(f)(W)]C_{V,W} F(f)(V_1)^* \oplus \ldots \oplus F(f)(V_k)^* C_{V,W}^{-1}
\end{equation}
For every pair of irreducible representations $V$, $W$ 
we choose an isomorphism
\[V \otimes W \simeq V_1
\oplus \ldots \oplus V_r\] with the $V_i$ not necessarily distinct
irreducibles. It follows that there are unitary matrices $C_{V,W}$
such that for all $f \in L^2(G)$
\[F(f)(V \otimes W) =
C_{V,W}\left( F(f)(V_1) \oplus  \ldots \oplus F(f)(V_r)\right) C^*_{V,W}
\] where we identify $F(f)(V \otimes W)$ as a matrix with respect to a pre-chosen
basis for $V \otimes W$.
Thus we can rewrite the bispectrum as
\begin{equation} \label{eq.biexplicit}
a_2(f)(V \otimes W) = [F(f)(V) \otimes F(f)(W)]C_{V,W}[F(f)(V_1)^* \oplus
\ldots \oplus F(f)(V_r)^*]C_{V,W}^*
\end{equation}

\begin{lemma} \label{lemma.key}
  Let $V$ and $W$ be irreducible representations of $G$
  and let $V_i$ be any irreducible appearing as a summand in $V \otimes W$.
  If $f \in L^2(G)$ is chosen such that the Fourier coefficients
  $F(f)(V)$ and $F(f)(W)$ are invertible, then the Fourier coefficient
  $F(f)(V_i)$ is determined by 
$F(f)(V)$, $F(f)(W)$
  and the coefficient $a_2(f)(V \otimes W)$ of the bispectrum.
\end{lemma}
\begin{proof}

Since
 \[a_2(f)(V \otimes W) =  [F(f)((V) \otimes F(f)(W)]
 [F(f)(V \otimes W)]^*\]
 and $F(f)(V) \otimes F(f)(W)$ is invertible by hypothesis, we obtain
 \[
  F(f)(V \otimes W)^* =  [F(f)(V) \otimes F(f)(W)]^{-1}a_2(f)(V \otimes W).\]
  Using the decomposition
  \[F(f)(V_1 \otimes V_2) = C_{V,W} [F(f)(V_1) \oplus \ldots \oplus F(f)(V_r)]C_{V,W}^*\] where $V_1, \ldots , V_r$ are irreducible yields the lemma.
\end{proof}
\begin{remark}
  The value of Lemma \ref{lemma.key} is that it shows that if $V,W$ are irreducibles and $V_i$ is an irreducible summand appearing in $V \otimes W$
  then $F(f)(V_i)$ is determined from $F(f)(V)$, $F(f)(W)$
  and the bispectrum coefficient $a_2(f)(V \otimes W)$. We will use this observation repeatedly.
  \end{remark}

\begin{proposition}\cite[Theorem 5]{smach2008generalized} 
  If the Fourier coefficients of $f \in L^2(G)$ are all non-singular
  then $f$ is uniquely determined by its bispectrum
\end{proposition}
\begin{proof}[Proof Sketch]
We first observe that for every irreducible representation $V$,
$a_2(f)(V \otimes {\bf 1}) = F(f)(V) F(f)(V)^*$
where ${\bf 1}$ denotes the trivial representation. Hence we know
the matrices $F(f)(V)$ up to multiplication by some unknown unitary matrix
$u(V)  \in U(V)$. In particular
if $h$ is a function with the same bispectrum as $f$
then $F(h)(V) = F(f)(V) u(V)$ for all irreducible representations
$V$. The goal is to show that these unitary matrices are all
of the form $D_V(g)$ for a fixed $g \in G$.

Since we know that the function $h$ whose  Fourier coefficient
$F(h)(V) = F(f)(V)u(V)$ has the same bispectrum
as $f$ we see that
\begin{multline} [F(f)(V) \otimes F(f)(W)] [(F(f)(V \otimes W))]^* = \\
  [F(f)(V)u(V) \otimes F(f)(W)u(W)][F(f)(V \otimes W)u(V \otimes W))]^*
\end{multline}
Since we assume that Fourier coefficients are all invertible we can conclude
that $u(V) \otimes u(W) = u(V \otimes W)$. Moreover, if
$A \colon V \to W$ is an intertwining operator between representations;
i.e. a $G$-invariant element of $\Hom(V,W)$ then $Au(V) = u(W)A$.

These facts imply that $u(V) = D_V(g)$ for some fixed element $g \in G$,
by Tannaka-Krein duality \cite[Theorem 30.43]{hewitt1970abstract}.
\end{proof}

\section{Banding functions for simple compact Lie groups}
The goal of this section is to introduce the notion of 
band-limited functions on compact Lie groups, generalizing the usual notion of band limited functions on $S^1$. We refer the reader to Appendix~\ref{sec.liegroup} for some of the basic terminology in the theory of compact Lie groups.

For functions on $S^1$ the notion of band-limiting is well understood.
We say that $f \in L^2(S^1)$ is $b$-band-limited if the Fourier coefficient
$f_n = \int_{S^1} e^{i n \theta} f(\theta)\;d\theta = 0$ for $|n| > b$ where
the functions $\{ e^{i n \theta}\}_{n \in \Z}$ form an orthonormal basis for $L^2(S^1)$. For functions on $S^2$ there is also a corresponding notion of
banding using the fact that any $f \in L^2(S^2)$ can be expanded in terms
of spherical harmonics $\{Y_{\ell}^m(\phi, \theta)\}$ where $\ell \in \N$ and
for each $\ell$, $m$ ranges from $-\ell$ to $\ell$. In this context
we say that $f = \sum_{\ell, m} a_{\ell, m} Y_\ell^m$ is $L$-band limited
if $a_{\ell,m} = 0$ for all $\ell > L$ and all $m \in [-\ell, \ell]$.
Band limited functions on $S^2$ can be understood in terms of the representation
theory of $\SO(3)$ as follows. The space of functions $L^2(S^2)$ decomposes
as a representation of $\SO(3)$ into an infinite sum of irreducibles
$\oplus_{\ell \geq 0} V_\ell$ where $V_\ell$ is the $(2\ell +1)$-dimensional
irreducible representation of $\SO(3)$ spanned by the functions $\{Y_\ell^m\}_{m=-\ell, \ldots \ell}$. With this notation any $f \in V_\ell$ has {\em band $\ell$} and the space of $L$-band-limited functions is the finite-dimensional
representation $\oplus_{\ell=0}^L V_\ell$.

Using the theory of highest-weight vectors we show that the irreducible
representations of a large class of Lie groups, including the classical
groups $\SU(n), \SO(2n), \Sp(n)$ can be banded. However, for groups
of rank more than one, there will be more than one representation with
a given band $b$. Nevertheless,
every irreducible representation of a given band $b$ appears as a summand
in the tensor product of irreducible representations of lower band.
As a consequence we can use a generalized
frequency marching argument to show (Proposition \ref{prop.bandone})
that if suitable Fourier coefficients
are invertible then the Fourier coefficients of irreducible representations
of band $b > 1$ can be recovered from the bispectrum and the Fourier coefficients of the irreducible representations of band one. For the classical groups,
$\SU(n), \SO(n), \Sp(2n)$ we can go further and prove
in Theorem \ref{thm.classical} that 
the Fourier coefficients of the irreducible representations of band one can be ordered in such a way that
they are determined from the bispectrum and the Fourier coefficient
of a single representation which we call the {\em defining representation}.
The defining representation corresponds to the smallest realization
of the particular group as a group of unitary matrices. For $\SU(n)$ and $\SO(n)$ it
is an $n$-dimensional representation, while for $\Sp(2n)$ it is a $2n$-dimensional representation.

%To define the general notion of band-limited function we need some terminology %from representation theory.

\subsubsection{Fundamental representations and banding for simply connected groups}
We begin with the case that $G$ is simply connected which is
true for $G = \SU(n)$ or $G=\Sp(n)$. Any compact Lie group
is the maximal compact subgroup of a corresponding complex algebraic
group $G_\C$ \cite[Propositions 8.3, 8.6]{brocker1995representations}.
When $G$ is simply
connected, the representations of $G$ correspond to the representations
of ${\mathfrak g}$ where ${\mathfrak g}$ is the Lie algebra of
$G_\C$ \cite[Theorem 3.7]{hall2015lie}.
Since ${\mathfrak g}$ is a semi-simple Lie algebra any
representation decomposes into a sum of irreducible representations.
Each irreducible representation decomposes as a sum of
{\em weight spaces}. These are the common eigenspaces for the
action of the Cartan subalgebra ${\mathfrak h}$ which
is the maximal abelian Lie subalgebra of ${\mathfrak g}$. The dimension
of ${\mathfrak h}$ is called the {\em rank} of ${\mathfrak g}$.
The eigenvalues for the action of ${\mathfrak h}$ on
all representations of ${\mathfrak g}$ generate a lattice
in ${\mathfrak h}^*$ called the weight lattice $\Lambda_W$.
The eigenvalues for the action of ${\mathfrak g}$ on itself are
called {\em roots} and they generate
the root lattice $\Lambda_R$ in ${\mathfrak h}^*$.
The root lattice $\Lambda_R$ has finite index in the weight lattice
$\Lambda_W$ and $\Lambda_R$ is the dual lattice to $\Lambda_W$
with respect to a natural inner product on ${\mathfrak h}^*$.
The set of roots $\Phi$ of ${\mathfrak g}$ can be divided (by choice
of a hyperplane in ${\mathfrak h}^*$) into positive and negative roots.
The {\em positive simple roots} form basis for the root lattice
characterized by the property that any positive root is a non-negative
integral linear combination of the positive simple roots.
A weight vector $\lambda \in \Lambda_W$ is {\em dominant}
if its inner product with every positive simple root is non-negative.
Any dominant weight vector is a
non-negative integral linear span of the {\em fundamental weights}
$\omega_1, \ldots , \omega_n$ where $n$ is the rank of the Lie
algebra ${\mathfrak g}$.

Irreducible representations of a semi-simple Lie algebra
are determined by their {\em highest weight vectors},
which is the unique dominant weight of the irreducible
representation which maximizes the sum of the inner products
with the fundamental weights. Moreover, any dominant weight 
$\lambda= a_1 \omega_1 + \ldots + a_n\omega_n$ with $a_i \in \N$
is the highest weight vector for a unique irreducible
representation $V_\lambda$ \cite[p.205]{fulton1991representation}.

\begin{defn}
When $G$ is simply connected we define the {\em band} of
the representation $V_\lambda$ with highest weight vector $\lambda = a_1\omega_1 + \ldots + \omega_n$
to be $b = a_1 + \ldots + a_n$.
\end{defn}
\subsubsection{Type $A_{n-1}$ - the group $\SU(n)$}
The compact group $\SU(n)$ is a compact form the algebraic group
$\SL(n,\C)$ and because $\SU(n)$ is simply connected
representations of $\SU(n)$ bijectively correspond
to representations of the complex Lie algebra $\mathfrak{sl}_{n}$ which has type
$A_{n-1}$ for $n \geq 2$. The Lie algebra $\mathfrak{sl}_n$ is the
vector space of traceless $n \times n$ complex matrices and the
Cartan subalgebra ${\mathfrak h}$ is the subspace of diagonal traceless
matrices. 

The weight lattice 
is the lattice spanned by vectors $L_1, \ldots , L_{n-1}, L_n$ with
$L_1 + \ldots + L_n =0$, where $L_i$ is the function on ${\mathfrak h}$
which reads the $i$-th entry along the diagonal.
In this case the positive simple roots are
\[ L_1-L_2, \ldots , L_{n-1} -L_n\]
and the fundamental weights
 are \cite[p. 216]{fulton1991representation}
\[
\omega_i=\sum_{j\leq i}L_j
\]
for $i<n$.
The irreducible representation corresponding to $\omega_1$ is the $n$-dimensional defining representation $V$;
i.e. the representation $\SU(n) \subset U(n)$. The other representations
of band one (i.e. those associated to the fundamental representations)
are the exterior powers $V_k = \wedge^k V$ for $i =1, \ldots , n-1$. (Note
that $\wedge^nV$ is the trivial representation
which is consistent with the fact that $L_1 + \ldots L_n =0$.)
\begin{example}[The groups $\SU(2)$, and $\SU(3)$]
 The group $\SU(2)$ has rank-one and
  there is a single representation
  of each band, namely the representation $V_n = \sym^n V_1$ where
  $V_1$ is the two-dimensional defining representation and $V_n = \sym^n V_1$ can be identified
  with the vector space of homogeneous binary forms of degree $n$.

  By contrast, the group $\SU(3)$ has two irreducible representations
  of band-one, the defining representation $V_1$ and $\wedge^2V_1 \simeq V_1^*$.
  The irreducible representations of band $b$ can be indexed by pairs of non-negative
  integers $(n_1, n_2)$ with $n_1 + n_2 = b$. Hence there are $b+1$ irreducible
  representations of each band. The dimension of the irreducible
  representation $\Gamma_{n_1,n_2}$ with highest weight vector
  $n_1\omega_1 + n_2 \omega_2$ is $\frac{(n_1 + 1)(n_2 +1)(n_1 + n_2 + 2)}{2}$
  \cite[Formula (15.17), p.~224]{fulton1991representation}. For example,
  the three representations of band two $\Gamma_{2,0}, \Gamma_{1,1}, \Gamma_{0,2}$
have dimensions $6, 8, 6$ respectively. Explicitly 
$\Gamma_{2,0} = \Sym^2 V_1$, $\Gamma_{0,2} = \Sym^2 V_1^*$
and 
$\Gamma_{1,1}$ is the kernel of the pairing $V_1 \otimes V_1^* \to \C$ defined by 
$v \otimes f \mapsto f(v)$. 
  \end{example}
\subsubsection{Type $C_n$- the group $\Sp(n)$}
The compact symplectic group $\Sp(n)$ is the intersection of
the complex symplectic group $\Sp(2n, \C)$ with the unitary group
$U(2n,\C)$. Since this group is simply connected the irreducible
representations of $\Sp(2n)$ are the same as the irreducible representations
of the complex Lie algebra $\mathfrak{sp}_{2n}$ which has type $C_n$ in
the classification of Lie algebras.
In this case the weight lattice is freely generated by vectors  $L_1,\dots,L_n$. The positive simple roots are
\[
L_1-L_2,\dots,L_{n-1}-L_n,2L_n
\]
and the fundamental weights are
\[
\omega_i=\sum_{j\leq i}L_j
\]
for $i\leq n$ \cite[Section 17.1]{fulton1991representation}.
There are $n$ irreducible representations of band one 
and the irreducible representation $V_k$ with highest weight vector
$\omega_k$ is the kernel
of the contraction map $\wedge^k V \to \wedge^{k-2}V$ \cite[Theorem 17.5]{fulton1991representation}.

\subsection{The non-simply connected groups $\SO(n)$}
As discussed in \cite[Section 23.1]{fulton1991representation}
a general simple Lie group $G$ has a finite abelian fundamental group,
so it is a quotient $\tilde{G}/Z$ where $\tilde{G}$ is simply connected
and $Z$ is a finite abelian group. Any
irreducible representation of $G$ is an irreducible representation of
$\tilde{G}$ but not every irreducible representation of $\tilde{G}$
descends to an irreducible representation of $G$. The set of
weights of representations of $G$ forms a sublattice of the set of
weights of representations of the universal cover $\tilde{G}$ of
finite index. In this case there need not be an analogue
of fundamental weights. That is, we cannot guarantee that the weight
lattice of $G$ has a basis $\omega_1, \ldots, \omega_n$  such that
every highest weight vector can be written as non-negative integral linear
combination of the $\omega_i$.
For the group $\SO(2n+1)$ it is possible to
find such weights, and thus we can define the band of an irreducible representation as above. For the group $\SO(2n)$ the weight lattice does not have fundamental system of weights. Despite this, we are still able to define the band
of an irreducible representation as we see below.

\subsubsection{Type $B_n$ - the group $\SO(2n+1,\R)$}
The compact group $\SO(2n+1,\R)$ is a compact form of the complex group
$\SO(2n+1,\C)$ with Lie algebra $\mathfrak{so}_{2n+1}$ which has type $B_n$.
The root system of type $B_n$ has weight space generated by $L_1,\dots,L_n$. The positive simple roots are \cite[Section 19.4]{fulton1991representation}
\[
\alpha_1=L_1-L_2,\dots,\alpha_{n-1}=L_{n-1}-L_n,\alpha_n=L_n
\]
and the fundamental weights are
\[
\omega_i=\sum_{j\leq i}L_j
\]
for $i<n$ and $\omega_n=\frac{1}{2}\sum_{j=1}^n L_j$.

Note that because $\SO(2n+1)$ is not simply connected not every irreducible
representation of $\mathfrak{so}_{2n+1}$ gives rise to an irreducible representation
of $\SO(2n+1)$. The representations of the Lie algebra $\mathfrak{so}_{2n+1}$ are in bijective correspondence with the representations of the simply connected
spin group $\Spin(2n+1)$ and the irreducible representations of $\SO(2n+1)$
are exactly the representations with highest weight vectors
$a_1\omega_1 +  \ldots  + a_{n-1}\omega_{n-1} + a_n \omega_n$ where $a_n$
is required to be even \cite[Proposition 23.13(iii)]{fulton1991representation}. In
particular we can take 
the vectors
$\omega_k = L_1 + \ldots + L_k$ for $k < n$
and $\omega'_n = 2\omega_n =L_1 + \ldots + L_n$ to be a set of fundamental weights for the Lie group $\SO(2n+1)$.

The representations $V_1, \ldots , V_n$ associated
to the fundamental weights are the exterior powers $V_1 = \wedge^1 V, \ldots ,
V_n = \wedge^nV$ where $V$ is the defining representation of $\SO(2n+1)$. Note that $V$ has dimension $2n+1$ \cite[Theorem 19.14]{fulton1991representation}.
%while the module with highest weight
%$\omega'_n$ is a submodule of $\wedge^nV$ (check this).
\begin{example}
  The group $\SU(2)$ is the universal cover of $\SO(3)$. Since
  these groups have rank one, the weights are integers. For $\SU(2)$ the fundamental weight is $\omega = 1$, while for $\SO(3)$ the fundamental weight
  is $\omega = 2$. As a result, the irreducible $\SO(3)$-representation
  of band $b$, which has dimension $2b+1$, is the same as the irreducible
  representation of $\SU(2)$ of band $2b$. Note that the weight lattice for $\SO(3)$ has index $2$ in the weight lattice for $\SU(2)$, corresponding to the fact that $\SU(2)\to\SO(3)$ is a $2$-to-$1$ cover.
\end{example}

\subsection{Type $D_n$ - the group $\SO(2n,\R)$} The group
$\SO(2n,\R)$ is the compact form of $\SO(2n,\C)$ whose
Lie algebra $\mathfrak{so}_{2n}$ has type $D_n$. Its weight space is generated by $L_1,\dots,L_n$. The positive simple roots are
\[
\alpha_1=L_1-L_2,\dots,\alpha_{n-1}=L_{n-1}-L_n,\alpha_n=L_{n-1}+L_n
\]
and the fundamental weights are \cite[Section 19.2]{fulton1991representation}
\[
\omega_i=\sum_{j\leq i}L_j
\]
for $i<n-1$, 
\[
\omega_{n-1}=\frac{1}{2}\sum_{j=1}^nL_j, \quad\textrm{and}\quad \omega_n=\frac{1}{2}(L_1+\dots+L_{n-1}-L_n).
\]
Once again, $\SO(2n)$ is not simply connected so not every irreducible representation of $\mathfrak{so}_{2n}$
gives rise to a representation of the Lie group $\SO(2n)$.
The irreducible representations of $\SO(2n)$ are precisely those with highest weight vector 
$\sum_{i=1}^na_i\omega_i$ where $a_{n-1}+a_n$ is even \cite[Proposition 23.13(iii)]{fulton1991representation}. 
In this case every highest weight vector can be expressed non-uniquely
as non-negative
linear combination of the weights
$\omega_1, \ldots , \omega_{n-2}$ and
\[ \omega'_{n-1} = \omega_{n-1} + \omega_n = L_1 + \ldots + L_{n-1}, \]
\[ \omega'_n = 2\omega_{n-1}= L_1 + \ldots + L_n,\]
\[ \omega'_{n+1} =2\omega_n = L_1 + \ldots + L_{n-1} -L_n.\]
If $\lambda$ is a highest weight vector for an irreducible
representation of $\SO(2n)$ and we write 
\[ \lambda = b_1 \omega_1 + \ldots + b_{n-2}\omega_{n-2} + b_{n-1}\omega'_{n-1}
+ b_{n}\omega'_n + b_{n+1}\omega'_{n+1}\]
then, although the non-negative integers
$b_i$ are not unique, the sum $\sum_{i=1}b_i$ is independent of the
choice of $b_i$'s. 
For example,
the weight $2L_1 + \ldots 2L_{n-1}$ has band 2 since it can be expressed
as $2\omega'_{n-1}$ or as $\omega'_n + \omega'_{n+1}$. Hence
we can define the band of $\lambda$ to be $\sum_{i=1}^n b_i$.
In particular, the $n+1$ irreducible representations with highest weights $\omega_1, \ldots , \omega_{n-2},
\omega'_{n-1}, \omega_n'$ all have band one.

By \cite[Remark page 289]{fulton1991representation}, if $k \leq n-2$ then
the representation $V_k$ with highest weight vector $\omega_k$ is the exterior product $\wedge^k V$. Likewise, the representation
$V_{n-1}$ with highest weight $\omega'_{n-1}$ is
exterior power $\wedge^{n-1} V$. Finally the exterior power $\wedge^nV$
is the sum of the representations $V_{n}$ and $V_{n+1}$ which have highest weights
$\omega'_n$ and $\omega'_{n+1}$ respectively.

\begin{example}[Irreducible representations of $\SO(4)$]
  The group $\SO(4)$ has rank two and there are three representations
  of band one.  The defining representation $V_1$ has dimension four
  and the two additional representations of band-one, 
  $V_2', V_3'$, are both three-dimensional. The sum  $V_2'\oplus V_3'$
  equals $\wedge^2 V_1$. If we choose an orthonormal
  basis $e_1, e_2, e_3, e_4$ for $V_1$ then $V_2'$ is the span
  of the vectors 
  $$e_1 \wedge e_2 + e_3 \wedge e_4,\quad e_1 \wedge e_4 + e_2 \wedge e_3,\quad
  e_1 \wedge e_3 - e_2 \wedge e_4$$ 
  in $\wedge^2 V_1$, while
  $V_3'$ is the span of 
  $$e_1 \wedge e_2 - e_3 \wedge e_4,\quad e_1 \wedge e_4 - e_2 \wedge e_3,\quad 
  e_1 \wedge e_3 + e_2 \wedge e_4$$ in $\wedge^2 V_1$.
\end{example} 

%\begin{definition}
%  The band of an irreducible representation $V_\lambda$ is
%  the sum $b(\lambda)  = a_1 + \ldots +a_n$. Note that the trivial representation has band $0$ and a representation has band one iff its highest weight vector is a fundamental weight.
\subsection{Reduction to representations of band one}

\begin{proposition} \label{prop.bandone}
  Let $G$ be a compact Lie group whose irreducible representations
  can be banded.  If $f \in L^2(G)$ and $W$ is an
irreducible representation of band $b>1$ then the Fourier coefficient
$F(f)(W)$ is determined from the bispectrum coefficients $a_2(f)(W_1
\otimes W_{b-1})$ and the Fourier coefficients of $F(f)(W_1)$,
$F(f)(W_{b-1})$ for some $W_1$ of band one and some $W_{b-1}$ of band
$b-1$, provided that these coefficients are invertible.
\end{proposition}
\begin{proof}
Since the band of $W$ is $b$, there are weight vectors
$\omega_1, \ldots , \omega_r$ of band one such that 
the highest weight vector of $W$ is
$b_1 \omega_1 + \ldots + b_r \omega_r$ with $b_i \in \N$ and $\sum_i b_i = b$.
Since $b > 1$ we know
that one of the $b_i$'s is positive. Take $W_1$ to be the
irreducible representation with highest weight $\omega_i$ and
$W_{b-1}$ to be the irreducible representation with highest weight
$b_1\omega_1 + \ldots + (b_{i}-1)\omega_{i} + \ldots + b_r \omega_r$. Since weights are additive on tensor products, $b_1 \omega_1 + \ldots + b_r \omega_r$ is the highest weight in the tensor product $W_1 \otimes W_{b-1}$. Hence
  $V$ is a summand in this tensor. Therefore by Lemma \ref{lemma.key}
  we conclude that $F(f)(W)$ is determined.
\end{proof}

\begin{corollary} \label{cor.halfband}
  With the hypotheses on $G$ as above, if $f \in L^2(G)$ is a function which is band-limited at band $b > 1$
  and the Fourier coefficients of all representations of band $0 \leq i \leq \lceil b/2 \rceil$ are invertible, then all Fourier coefficients can be determined
  from the Fourier coefficients of band one and the bispectrum.
\end{corollary}
\begin{proof}
Let $W$ be an irreducible representation with highest weight vector
$ \lambda = b_1 \omega_1 + \ldots + b_r \omega_r$
where the $b_i$ are non-negative integers and the $\omega_i$ have band one.
We can decompose the vector
  $(b_1, \ldots, b_r) = (c_1, \ldots ,c_r) + (d_1, \ldots , d_r)$
with $b_i,c_j$ non-negative integers and
$\sum c_i, \sum d_i \leq \lceil b/2 \rceil $. If
$W_1$ is the irreducible represention with highest weight
$c_1 \omega_1 + \ldots c_r \omega_r$ and $W_2$ is
the irreducible representation with highest weight $d_1 \omega_1 + \ldots + d_r \omega_r$, then
$W$ appears as a summand in the tensor product $W_1 \otimes W_2$.
Thus, if $F(f)(W_1)$ and $F(f)(W_2)$ are invertible then $F(f)(W)$
can be determined from $a_2(f)(W_1 \otimes W_2)$ and $F(f)(W_1)$, $F(f)(W_2)$ 
regardless of whether $F(f)(W)$ is invertible.
By Proposition \ref{prop.bandone} and induction $F(f)(W_1)$, $F(f)(W_2)$ can be determined from the Fourier coefficients of band one.
\end{proof}

\subsection{Reduction to the defining representation}
The previous propositions only required that the irreducible representations of
our compact Lie group $G$ be 
banded. We now state a result specific to the
classical groups $\SU(n), \SO(2n+1), \Sp(n), \SO(2n)$.
\begin{theorem} \label{thm.classical}
  Let $G$ be one of the classical compact Lie groups $\SU(n), \SO(n), \Sp(2n)$
  with defining representation $V_1$.
  If the Fourier coefficients of the irreducible representations
  of band one are all invertible, then the Fourier coefficient $F(f)(V_\ell)$ 
  of any band-one representation $V_\ell$
  is determined by $F(f)(V_1)$ and 
  the bispectrum matrices $a_2(f)(V_1 \otimes V_k)$ with $V_k$ of band-one
  and $k < \ell$.
\end{theorem}
\begin{proof}
The proof is based on a case-by-case analysis, but the overall structure is the same in each case and makes use of a simple lemma.
\begin{lemma} \label{lem.wedge}
Let $V$ be a representation of a compact Lie group $G$. Then $\wedge^{k+1}V$
appears as a $G$-invariant summand in $V \otimes \wedge^kV$.
\end{lemma}
\begin{proof}
  The action of $G$ on $V$ defines a homomorphism $G \to U(V)$ where $U(V)$
  is the group of unitary transformations of $V$. In particular, it suffices to prove the lemma when $G$ is the unitary group $U(d)$ where $d = \dim V$. The statement then follows from the fact that the map
  $V \otimes \wedge^kV \to \wedge^{k+1}V$, defined by 
  \[v \otimes (v_1 \wedge \ldots \wedge v_k) \mapsto v \wedge v_1 \ldots \wedge v_k\]
  is surjective and commutes with the respective actions of the unitary group $U(V)$ on $V \otimes \wedge^k V$ and $\wedge^{k+1} V$.
\end{proof}

\subsubsection{Type $A_{n-1}$ - the group $\SU(n)$}
As noted above the representation 
$V_k$ with highest weight
$\omega_k = L_1 + \ldots + L_k$ is $\wedge^kV_1$ where $V_1$ is the defining
representation. 

By Lemma \ref{lemma.key} and induction it suffices to show that $V_{k+1}$ is a
summand in the tensor product
$V_1 \otimes V_k$ which follows from Lemma \ref{lem.wedge}.

\subsubsection{Type $B_n$ - the group $\SO(2n+1)$}
Here we have a fundamental system of weights $\omega_1, \ldots ,\omega_{n-1}, \omega'_n$ and the corresponding irreducible representations
are the exterior product $\wedge^k V_1$ for $1 \leq k \leq n$.
Once again by Lemma \ref{lemma.key} we just need to show
that $V_{k+1}$ is a summand in $V_1 \otimes V_{k}$ which follows from Lemma \ref{lem.wedge}.

\subsubsection{Type $C_n$- the group $\Sp(2n)$}
Unlike the case of $\SU(n)$ and $\SO(2n+1)$ the irreducibles associated to the fundamental weights are not exterior powers of the defining representation $V_1$.
However, \cite[Theorem 17.5]{fulton1991representation} states
that if $k > 1$ then
$V_k$ is the kernel of the contraction map $\wedge^k V_1 \to \wedge^{k-2}V_1$.
Hence $\wedge^kV_1 = V_k \oplus V_{k-2} \oplus \ldots \oplus  V_0$ if $k$ is even
and $\wedge^kV_1 = V_k \oplus V_{k-2} \oplus \ldots \oplus V_1$ if $k$ is odd.
Since the Fourier coefficient of the trivial representation $V_0 = \wedge^0 V_1$
is known from the bispectrum by Proposition \ref{prop.3to2}, we are able to inductively determine
the Fourier coefficient $F(f)(V_k)$ from the bispectrum and knowledge of $F(f)(V_1)$ as follows: Assume by induction that we have determined the Fourier
coefficients $F(f)(V_\ell)$ for $\ell \leq k$. Since $\wedge^k(V_1) = \oplus_{\ell} V_{k-2\ell}$ we know the Fourier  $F(f)(\wedge^kV_1)$ by induction. Since
$\wedge^{k+1}V_1$ appears as a summand in $V_1 \otimes \wedge^{k}V_1$ and
$V_{k+1}$ is a summand in $\wedge^{k+1}V$ we see that $V_{k+1}$ is a summand
in $V_1 \otimes \wedge^{k}V$. Since the Fourier coefficient $F(f)(\wedge^kV_1)$
is known by induction, it follows from Lemma \ref{lemma.key} that $F(f)(V_{k+1})$
is determined by $F(f)(V_1)$ and the bispectrum coefficient
$a_2(f)(V_1 \otimes \wedge^kV_1)$.

\subsubsection{Type $D_n$ - the case of $\SO(2n)$}
In this case we have $n+1$ weights of band one
$\omega_1, \ldots , \omega_{n-2}, \omega'_{n-1}, \omega'_{n},\omega'_{n+1}$
and associated representations $V_1, \ldots , V_{n+1}$. As noted above,
$V_k = \wedge^k V_1$ for $1 \leq k \leq n-1$. Hence the same argument used in
the case of $\SO(2n+1)$ implies that the bispectrum and
$F(f)(V_{1})$ determine $F(f)(V_k)$ if $k \leq n-1$.
However we also know that $V_{n}$ and $V_{n+1}$ are summands in $\wedge^n V_1$
so we can determine them from $F(f)(V_{n-1}), F(f)(V_1)$ and the
bispectrum coefficient $F(f)(V_1 \otimes V_{n-1})$.
\end{proof}

\section{Sharp results for $\SO(2n+1), \SU(n)$}
For the groups $G= \SU(n)$ and $G=\SO(2n+1)$ we prove
that the $G$-orbit of a generic 
band limited function in $L^2(G)$ is determined by its by bispectrum.
\begin{theorem} \label{thm.suso}

(i)  If $f \in L^2(\SU(n))$ is band limited with band $b \geq 1$
  and all Fourier coefficients of irreducible representations
  whose bands are at most $\lceil b/2 \rceil$ are invertible then the $\SU(n)$ orbit of $f$ is determined by its bispectrum. 

(ii) If $f\in L^2(\SO(2n+1))$ is real valued and band limited with
    band $b \geq 1$
and all Fourier coefficients of irreducible representations
  whose band is at most $\lceil b/2 \rceil$ are invertible
  then the $\SO(2n+1)$
    orbit of $f$ is determined by its bispectrum.
\end{theorem}
\begin{proof}
  By Theorem \ref{thm.classical} we know that Fourier coefficients of all irreducibles
  are determined by the Fourier coefficient $F(f)(V)$ of the defining representation. Moreover, by
  Propositions \ref{prop.secondmoment} and \ref{prop.3to2}, we also know $F(f)(V) F(f)(V)^*$ so
  we know $F(f)(V)$ up to translation by an element of $U(V)$ if $f$ is complex valued and $O(V)$ if $f$ is
  real-valued.

  Suppose that $f' \in L^2(\SU(n))$ has the same bispectrum as $f$. Then
  we know that $F(f')(V) = uF(f)(V)$ for some $u \in U(V)=  U(n)$. Our goal
  is to show that $u \in \SU(n)$. Any element of $U(n)$ can be factored
  as $u = (\diag{e^{i\theta}} )r$ where $r \in \SU(n)$. Replace
  $f$ with the function whose Fourier coefficient of $V$ is $rF(V)$ and
  has the same bispectrum. (Such a function exists because the bispectrum is invariant under the action of $\SU(n)$ on $L^2(\SU(n))$.) By doing so, we may
  reduce to the case that $F(f')(V) = e^{i \theta}F(f)(V)$
  and has the same
  bispectrum. Our goal is to show that $\diag{e^{i \theta}} \in \SU(n)$ or, equivalently,
  that $e^{i \theta}$ is a $n$-th root of unity.

  Since the bispectra of $f$ and $f'$ are equal we see that 
  \begin{eqnarray*}
    a_2(f)(V \otimes V) & = & [F(f)(V) \otimes F(f)(V)]F(f)(V \otimes V)^* \\
    & = & [F(f')(V) \otimes F(f')(V)]F(f')(V \otimes V)^*\\
    & = & e^{2i \theta} [F(f)(V) \otimes F(f)(V)]F(f')(V \otimes V)^*
  \end{eqnarray*}
  where the last equality follows from the fact that
  $\diag{e^{i \theta}} \otimes \diag{e^{i \theta}}$ is $e^{2i\theta}$ times
    the identity operator on $V \otimes V$.

  If follows that for any irreducible $W$ appearing in $V \otimes V$,
  we have that $F(f')(W) = e^{2i \theta} F(f)(W)$. In particular
  if $V_2 = \wedge^2 V$ then $F(f')(V_2) = e^{2i\theta}F(f)(V_2)$.
  Continuing this way we see that for the fundamental representations
  $V=V_1, \ldots V_{n-1}$, $F(f')(V_k) = e^{i k \theta}F(f)(V_k)$.
    On the other hand we know that the bispectrum uniquely determines the
    Fourier coefficient of the trivial representation $V_0$, so we must have that $F(f')(V_0) = F(f)(V_0)$. However, $V_0$ appears as a summand
    in $V \otimes V_{n-1}$ so by our previous argument we see that
    $F(f')(V_0) = e^{i n \theta} F(f)(V_0)$. Therefore $e^{i n \theta} = 1$
    as desired.

    The proof for $\SO(2n+1)$ is similar to the proof for $\SU(n)$ but requires
    a slightly more complicated representation-theoretic argument.
    If $f \in L^2(\SO(2n+1))$ is real-valued and $f'$ is another real valued function with the same bispectrum then we know that $F(f')(V) = oF(f)(V)$ where
    $o \in O(2n+1)$. Now any element in $O(2n+1)$ can be written 
    as $\pm r$ where $r \in \SO(2n+1)$.
    We will prove the the result by showing that if $o \in O(2n+1) \setminus \SO(2n+1)$
    we obtain a contradiction. Assuming that
    $o \notin \SO(2n+1)$ we can reduce to the case that 
    that $F(f')(V) = -F(f)(V)$ and $f'$, $f$ have the same bispectrum.
    The tensor product $V \otimes V$ contains the fundamental representation
    $V_2 = \wedge^2V$ as a summand, so we see that $F(f')(V_2) = (-1)^2F(f)(V)$.
    Continuing this way we see that $F(f')(V_k) = (-1)^k F(f)(V)$ for
    $1 \leq k \leq n$.
    Since $V_n = \wedge^n V$ we know by Lemma \ref{lem.wedge}
    that $V \otimes V_n$ contains a copy of $\wedge^{n+1}V$. Since
    $\dim V = 2n+1$, the exterior products $\wedge^{n}V$ and $\wedge^{n+1}V$
    are dual representations. However, we also know that as $\SO(2n+1)$
    representations, $\wedge^k V$ is self-dual for $k \leq n$. Hence
    $V \otimes V_n$ contains a copy of $V_n$.
This implies
    that $F(f')(V_n) = -F(f')(V_n)$, which is a contradiction.
\end{proof}
\section{Examples and Applications}
\subsection{Counterexamples}
  We give two examples, one for $S^1= \SO(2)$ and one for $\SU(2)$, that illustrate 
  the necessity of the hypothesis in Theorem~\ref{thm.suso} that the Fourier coefficients of band at most
  $\lceil \frac{b}{2} \rceil$ be invertible.
The reason we restrict to these groups is
that they are both rank one which makes the calculations more tractable.

\subsubsection{$S^1$ counterexample}
If $k \in \Z$ is any integer, denote by $V_k$ the one-dimensional
representation of $S^1$ where $e^{\iota \theta}$ acts on $V_k$ by
scalar multiplication by $e^{\iota k \theta}$.
With this notation, for any $\ell >0$ there are two one-dimensional representations of $S^1$ of band $\ell$, namely
$V_\ell$ and $V_{-\ell}$.

Consider the band-limited representation $W_3 = V_0 \oplus V_1 \oplus V_2 \oplus V_3$. Since every irreducible representation
of $S^1$
is one-dimensional the Fourier coefficients are scalars, so a Fourier
coefficient is invertible if and only if it is non-zero.
If $f \in W_3 \subset L^2(S^1)$ is a function, then the Fourier coefficient of $V_\ell$ is the
scalar $a_\ell$ in the Fourier expansion $f = \sum_{\ell=0}^3 a_\ell e^{i\ell \theta}$. As noted
in Example~\ref{ex.fourierband} if the Fourier coefficient of $a_2$ is zero
then we cannot recover the Fourier coefficient $a_3$ from
the bispectrum.

\subsubsection{$\SU(2)$ counterexample}
In this example we identify the defining representation
$V_1$ of $\SU(2)$ with the two-dimensional vector space of binary linear forms. The single irreducible representation $V_\ell$ of
band $\ell$ is $\Sym^\ell V_1$ and can be identified with the 
$(\ell +1)$-dimensional vector space of homogeneous binary forms of degree $\ell$.

Consider the 3-band limited representation of $\SU(2)$,
$W_3 = V_0 \oplus V_1 \oplus V_2 \oplus V_3$. Unlike the case for
$S^1$, the dimension of $V_\ell$ depends on $\ell$ as it has dimension
$\ell +1$. As a result the Fourier coefficient of
$V_\ell$ is not a scalar but an $(\ell+1) \times (\ell +1)$ matrix.
Let $f \in W_3 \subset L^2(\SU(2))$ be the function whose
Fourier coefficients are $F(f)(V_0) = 1, F(f)(V_1) = \Id_2,F(f)(V_2) = 0$
and $F(f)(V_3) =\Id_4$
where $\Id_k$ indicates the $k \times k$ identity matrix.
Using formula~\eqref{eq.fourierinv} for the inverse Fourier
transform we can explicitly compute $f$ as the function
$$A \mapsto 1 + 2\Tr A^{-1} + 4\Tr \sym^3A^{-1}.$$ If we write
$A = \begin{pmatrix}\alpha & \beta\\ -\overline{\beta} & \overline{\alpha}
  \end{pmatrix}$ 
  with $|\alpha|^2 + |\beta|^2 =1$
  then we have the explicit formula
$$A \mapsto 1+ 2(\alpha + \overline{\alpha}) +  4(\alpha + \overline{\alpha})(\alpha^2 + \overline{\alpha}^2 - 2|\beta|^2).$$  

 Since $f$ has only three non-zero Fourier coefficients $F(f)(V_0),
 F(f)(V_1), F(f)(V_3)$ the bispectrum has at most nine non-zero Fourier
 coefficients $a_2(f)(V_i \otimes V_j)$ for $i,j \in \{0,1,3\}$.
 The bispectrum is also symmetric; i.e., 
 $a_2(f)(V_i \otimes V_j)= a_2(f)(V_j \otimes V_i)$
 so we need only compute the six Fourier coefficients $a_2(f)(V_i \otimes V_j)$ with $i \leq j$.
 
 Because we have chosen the non-zero Fourier coefficients of $f$ to be the identity matrices, the three Fourier coefficients of
 $a_2(f)(V_0 \otimes V_j)$ are readily calculated using formula~\eqref{eq.bispectrum} and are:
 \begin{align*} a_2(f)(V_0 \otimes V_0) &= 1\\
 a_2(f)(V_0 \otimes V_1) & = \Id_2 \Id_2^* = \Id_2\\
 a_2(f)(V_0 \otimes V_3) & = \Id_4 \Id_4^* = \Id_4
 \end{align*}

 We also claim that $a_2(V_1 \otimes V_3)$ is the $8 \times 8$ zero matrix.
 The reason is that the tensor product $V_1 \otimes V_3$ decomposes
 as $V_4 \oplus V_2$ and the Fourier coefficients $F(f)(V_4)$ and
 $F(f)(V_2)$ are zero. Thus, $F(f)(V_2 \otimes V_3) = 0$. Hence
$$a_2(V_1 \otimes V_3) = [F(f)(V_1) \otimes F(f)(V_3)][F(f)(V_1 \otimes V_3)] = 0.$$

We now compute the Fourier coefficients $a_2(f)(V_1 \otimes V_1)$
and $a_2(f)(V_1 \otimes V_3)$.

The representation $V_1 \otimes V_1$ is the vector space of 
forms $q(x_0,x_1,y_0,y_1)$ which are homogeneous of degree one in
$(x_0,x_1)$ and $(y_0, y_1)$ respectively. This representation is isomorphic to the sum of the two irreducibles $V_2 \oplus V_0$. 
The
summand $V_0$ the one-dimensional subspace generated by
the $\SU(2)$-invariant form $x_0y_1 + x_1 y_0$. 
Since the Fourier coefficient $F(f)(V_2) =0$, the Fourier
coefficient $F(f)(V_1 \otimes V_1)$ is rank-one projection $P_1$
which projects $V_1 \otimes V_1$ to the one-dimensional subspace
spanned by the binomial $x_0y_1 + x_1y_0$. Note that $P_1 = P_1^*$
because the form $x_0y_1 + x_1 y_0$ is symmetric in the $x$ and $y$ variables.
Applying formula~\eqref{eq.bispectrum}
we see that 
\begin{eqnarray*}
a_2(f)(V_1 \otimes V_1) & = & [F(f)(V_1) \otimes F(f)(V_1)][F(f)(V_1 \otimes V_1)]^* \\
& = & [\Id_2 \otimes \Id_2]P_1^*\\
& = & P_1
\end{eqnarray*}
 
The remaining Fourier coefficient
is $$a_3(f)(V_3 \otimes V_3) = [\Id_4 \otimes \Id_4][F(f)(V_3 \otimes V_3)]^*.$$
The representation $V_3 \otimes V_3$ is isomorphic to the sum
$V_6 \oplus V_4 \oplus V_2 \oplus V_0$. The only non-zero Fourier
coefficient in this sum is that of the trivial representation $V_0$. Viewing
$V_3 \otimes V_3$ as the vector space of forms $s(x_0,x_1,y_0,y_1)$ which
are homogeneous of degree three in $(x_0,x_1)$ and $(y_0,y_1)$ respectively,
the invariant subspace $V_0$ is spanned by the form $(x_0y_1 - x_1y_0)^3$, and the Fourier coefficient $F(f)(v_2)(V_3 \otimes V_3)$
is the rank-one projection $P_3$ on to this subspace. Because
the form $(x_0y_1 - x_1y_0)^3$ is skew-symmetric in the $x$ and $y$
variable, $P_3^* = - P_3$.
Thus, by formula~\eqref{eq.bispectrum}
we see that $a_2(f)(V_3 \otimes V_3) = P_3^*=-P_3$

It is easy to construct functions $f'$ with the same bispectrum as
$f$ which are not in the same $\SU(2)$ orbit. The simplest example
is the functions whose Fourier coefficients are
$F(f')(V_0) = 1, F(f')(V_1) = \Id_2, F(f')(V_2) = 0, F(f')(V_3) = -\Id_3$
corresponding to the function on $\SU(2)$
defined by the formula
$$A \mapsto 1+ 2(\alpha + \overline{\alpha}) -  4(\alpha + \overline{\alpha})(\alpha^2 + \overline{\alpha}^2 - 2|\beta|^2).$$
More generally if $U$ is any $4 \times 4$ unitary matrix acting on
$V_3$ such that $U \otimes U \circ P_3 = P_3$ where $P_3$ is the projection
onto the subspace spanned by $(x_0y_1 -x_1y_0)^3$ then
the function $f'$ with generalized Fourier coefficients
$F(f')(V_0) = 1, F(f')(V_1) = \Id_2, F(f')(V_2) =0, F(f')(V_3) = U$
will have the same bispectrum as $f$.

\subsection{Algorithmic aspects} \label{sec.algorithms}
Although our result is theoretical, the proof of Theorem~\ref{thm.classical}
gives a potential algorithm for determining the Fourier coefficients of
all irreducible representations of a band-limited function from the Fourier
coefficient of the defining representation and the bispectrum. Moreover,
if all non-zero generalized Fourier coefficients of a band-limited
function are invertible, then the potential algorithm needs only a relatively
small part of the bispectrum to compute the unknown function
from the matrix Fourier coefficient of the defining representation.

Precisely, if $f \in L^2(G)$ is band-limited with band $b$
and $W_1, \ldots , W_t$ are the irreducible representations of band-one then
Proposition~\ref{prop.bandone} implies that we only need as 
as input, the values (which are matrices) of the bispectrum at
$\{W_i \otimes V\}$ where $V$ runs through all irreducible representations
of band strictly less than $b$. By contrast, if $f$ is $b$ band-limited
then the full bispectrum is determined by its values at $W \otimes V$ where
$W,V$ run over all irreducible representations of band at most $b$. As in the proof of Theorem~\ref{thm.classical} we proceed by determining the Fourier coefficients of band $\ell$ from the Fourier coefficients of band $\ell-1$ by performing at most $B_{\ell-1}$ matrix inversions and multiplications, where $B_{\ell-1}$ is the number of irreducible representations of band $\ell-1$. (Note that this strategy would also require knowing the decomposition of the tensor products
$W_i \otimes V$ into irreducible representations.) An interesting question for further work is to investigate
the robustness and stability of this approach from an algorithmic perspective.

For example, if $G = \SU(3)$ the number of irreducible representations
of band $\ell$ is $\ell +1$, so the bispectrum of a $b$ band-limited
function depends on its values at $\binom{b+1}{2}^2$ pairs of irreducible representations. However, in order to determine the Fourier
coefficients from the Fourier coefficient of the defining representation
we need only consider $2 \binom{b+1}{2}$ values of the bispectrum.
\subsection{Results for the group $\SO(3)$} \label{sec.so3}
For the classical Lie groups, the number and dimensions of the irreducible
representations of a given band grows exponentially in the rank of
of the group. However, for a group of rank-one such as $\SU(2)$ or $\SO(3)$
the computations may be feasible. For $\SO(3)$ there is a single
irreducible representation of band $\ell$, the representation $V_\ell$
of dimension $2\ell +1$ which has a basis of spherical harmonic polynomials
of $Y_{\ell,m}$ for $m = -\ell , \ldots , \ell$. In this case,
the coefficient of $V_1 \otimes V_{b-1}$ in the bispectrum is a $(6b-3) \times (6b-3)$ matrix. In particular, this shows that we can determine the Fourier coefficients of a $b$ band-limited function from the Fourier coefficient
$F(f)(V_1)$ and a polynomial in $b$ number of matrix multiplications and inversions. This suggests that for a group of rank one like $\SO(3)$ bispectrum
inversion of a band-limited function can be done in polynomial time in
the band.
\subsubsection{Comparison with the work of Liu and Moitra \cite{liu2021algorithms}} \label{sec.moitra}
A related result, with precise error bounds was proved by Liu and Moitra
\cite{liu2021algorithms}
for the MRA problem of $\SO(3)$ acting on band-limited functions on $S^2$.
Note that the space of $b$ band-limited functions on $S^2$ is isomorphic to the
sum of the representations $V_0 \oplus V_1 \ldots \oplus V_b$ where
$V_\ell$ is the $2\ell +1$-dimensional irreducible representation
of $\SO(3)$. This is in the contrast to the case for functions on $\SO(3)$
where the summand $V_\ell$ appears with multiplicity equal to its dimension
$2\ell +1$. In particular, their algorithm determines $1 + 3 + \ldots + (2b+1) = (b+1)^2$ unknown coefficients $f_{\ell m}$ coming from
the expansion of a band-limited function in spherical harmonics
as
$$\sum_{\ell \leq b} \sum_{m = -\ell}^{\ell} f_{\ell m} Y_{\ell m}(\theta, \phi).$$
By comparison, our goal is to determine a collection 
of matrices (the matrix Fourier coefficients). The main
result of \cite{liu2021algorithms} is a robust quasi-polynomial time algorithm which uses
the degree three
invariants together with knowledge of all coefficients
$f_{\ell m}$ with $\ell \leq C$ to determine the remaining coefficients
by frequency marching. (Here $C$ is a fixed constant independent of the band limit.)
It should be noted that it is an open theoretical problem as to whether
all unknown coefficients $f_{\ell m}$ can be determined solely from the degree-three invariants \cite[Section 4.5]{bandeira2017estimation}.

\subsection{Unprojected cryo-EM} \label{sec.cryo}
Let $L^2(\R^3)$ be Hilbert space of complex valued $L^2$ functions
on~$\R^3$. The action of $\SO(3)$ on $\R^3$ induces a corresponding
action on $L^2(\R^3)$, which we view as an
infinite-dimensional representation of $\SO(3)$.
In cryo-EM we are interested in the action of $\SO(3)$ on
the subspace of $L^2(\R^3)$ corresponding to the Fourier 
transforms of real valued functions on $\R^3$,  representing  the Coulomb  potential of an unknown molecular structure. 

Using spherical coordinates $(\rho, \theta, \phi)$, we consider a finite dimensional
approximation of 
$L^2(\R^3)$ by discretizing $f(\rho, \theta, \phi)$
with $R$ samples $r_1, \ldots , r_{R}$, of the radial coordinates
and band limiting the corresponding spherical functions $f(r_i, \theta, \phi)$. 
This is a standard assumption in the cryo-EM literature, see for example~\cite{bandeira2020non}.
Mathematically, this means that we approximate the infinite-dimensional representation $L^2(\R^3)$ with
the finite dimensional representation
$V = (\oplus_{\ell =0}^L V_\ell)^R$, where
$L$ is the band limit, and $V_{\ell}$ is the $(2\ell +1)$-dimensional
irreducible representation of $\SO(3)$, corresponding to harmonic
polynomials of frequency $\ell$.
An orthonormal basis for $V_\ell$
is the set of spherical harmonic polynomials $\{Y_\ell^m(\theta,
\phi)\}_{m = -\ell}^\ell$. We use the notation $Y_\ell^m[r]$
to consider the corresponding spherical harmonic as a basis vector for
functions on the $r$-th spherical shell. The dimension of this
representation is $R(L^2 + 2L+1)$.

Viewing an element of $V$ as a radially discretized function on $\R^3$, 
we can view
$f \in V$ as an $R$-tuple
$$f = (f[1], \ldots , f[R]),$$
where
$f[r] \in L^2(S^2)$ is an $L$-bandlimited function.
Each $f[r]$ can be expanded in terms of the basis functions $Y_\ell^m(\theta, \varphi)$ as follows 
\begin{equation} \label{eq.function}
  f[r] = \sum_{ \ell=0}^L\sum_{m=-\ell}^\ell A_{\ell}^m[r]
  Y_{\ell}^m.
  \end{equation}
Therefore, the problem of determining a structure reduces to determining the unknown coefficients $A_\ell^m[r]$ in \eqref{eq.function}.

Note that when
$f$ is the Fourier transform of a real valued function,  the coefficients
  $A_\ell^m[r]$ are real for even $\ell$ and purely imaginary for odd $\ell$~\cite{bhamre2015orthogonal}. 

For the case of $\SO(3)$ we can combine our results and obtain the following
corollary.
\begin{corollary} \label{cor.cryo}
  If $L \geq 3$ and any $R \geq L+2$ the generic orbit
  in the $\SO(3)$ representation $V = (\oplus_{\ell =0}^L V_\ell)^R$
  can be recovered from the third moment.
\end{corollary}
\begin{proof}
  By Proposition \ref{prop.bootstrap} it suffices to prove
  the corollary when $R = L+1$.
  Consider the $\SO(3)$-module
  $$W = V_0 \oplus V_1^3 \oplus \ldots V_{\lceil L/2 \rceil}^{2\lceil L/2 \rceil +1}
  \oplus V_{\lceil L/2 \rceil + 1}^{L+2} \oplus \ldots V_{L}^{L+2}.$$
  Since $L+2 \geq \lceil L/2 \rceil + 1$,
  we can once more invoke
  Proposition \ref{prop.bootstrap} and prove the result for the representation
  $W$.
We view $W$ as an $\SO(3)$-submodule of the vector space of
    $L$-bandlimited functions in $L^2(\SO(3))$ since the latter representation
    is isomorphic to $\oplus_{\ell =0}^L V_\ell^{2L+1}$. The generic element
    of $W$ viewed as a submodule of $\oplus_{\ell =0}^L V_\ell^{2\ell+1}$
    has invertible Fourier coefficients up to band $\lceil L/2 \rceil$.
    Therefore, by Corollary \ref{cor.halfband} if $f \in W$ then the
    Fourier coefficients
    $F(f)(V_\ell)$ are determined by the bispectrum
    and the single Fourier coefficient $F(f)(V_1)$ of band one.
     By Proposition \ref{prop.3to2}
    the third moment determines the second moment and by Proposition \ref{prop.secondmoment} the second moment determines $A_1(f)A_1(f)^T$ where
    $A_1(f)$ is the full-rank real $3\times 3$ matrix $(-iA_1^m[r])_{-1 \leq m \leq 1, 1 \leq r \leq 3}$. (Note that the coefficients $A_1^m[r]$ are purely imaginary so
    we multiply by $-i$ to obtain a real valued matrix.)
    In particular the matrix $A_1(f)$ (which we can identify with the Fourier coefficient in a suitable basis) is determined up to multiplication
    by an element $O(3)$. By Theorem \ref{thm.suso} $A_1(f)$ is determined
    up rotation by an element of $\SO(3)$. Hence the orbit of the generic
    vector $f \in W$ is determined from its third moment. 
\end{proof}
\begin{remark} An an analogous result holds for the groups $\SO(2n+1)$ and $\SU(n)$. However, because there are many representations of a given band and their dimensions vary, it cannot be stated as precisely as the corresponding statement for $\SO(3)$. The general statement is the following: Let  $R$ be at least $\max \dim V$ where $V$ runs over all irreducible representations of band $\lceil b/2 \rceil$
and let $W = \oplus_{\{V| \band V \leq b\}}V^R$. Then the generic $\SO(2n+1)$ (resp. $\SU(n)$) orbit in $W$ can be recovered from the third moment.
\end{remark}
\subsubsection{Further directions} In
  \cite{bandeira2017estimation} it is conjectured that for any $L$ the
  third moment separates generic orbits provided $R \geq 3$. This
  conjecture was verified for $1 \leq L \leq 15$ using techniques from
  computational commutative algebra and a computational algebra
  algorithm was presented for recovering the orbit using frequency
  marching.  
  Note that Corollary \ref{cor.cryo} is
  weaker than the conjectured bounds of \cite{bandeira2017estimation}
  in that we require the multiplicities of the irreducible
  representations to grow with the band. An interesting direction for
  further work is to refine the the methods used here to determine
  whether the third moment carries enough information to separate
  orbits when the irreducible representations have constant
  multiplicity which is independent of the band limit.
  
There is an expectation in the cryo-EM community that the generic orbit can be recovered from the projected third moment. In \cite[Section 4.5]{bandeira2017estimation} the authors have computationally verified that the projected third moment recovers generic orbits up to a finite list (list recovery). An important problem is to mathematically prove that the projected third moment
separates generic orbits in the spherical shells model. We view Corollary~\ref{cor.cryo} as a first step in this direction - particularly since we can recover generic orbits from a very small portion of the information carried by the bispectrum/third moment.

Another interesting avenue of investigation is the case of finite groups. It is known
\cite{bandeira2017estimation} that the third moment separates generic orbits in any representation containing the regular representation. However, there are essentially no known non-trivial examples of smaller representations of finite groups where the third moment separates generic orbits.
\section*{Acknowledgment}
We thank Tamir Bendory for helpful discussions.
The work of D.E. was supported by the BSF grant no. 2020159 and NSF-DMS 2205626. The work of M.S. 
was partially supported by a Discovery Grant from the National Science and Engineering Research Council of Canada as well as a Mathematics Faculty Research Chair from the University of Waterloo.

\bibliographystyle{siamplain}
%\bibliography{../ref}

\begin{thebibliography}{10}

\bibitem{abbe2018multireference}
Emmanuel Abbe, Tamir Bendory, William Leeb, Jo{\~a}o~M Pereira, Nir Sharon, and
  Amit Singer.
\newblock Multireference alignment is easier with an aperiodic translation
  distribution.
\newblock {\em IEEE Transactions on Information Theory}, 65(6):3565--3584,
  2018.

\bibitem{bandeira2017estimation}
Afonso~S Bandeira, Ben Blum-Smith, Joe Kileel, Amelia Perry, Jonathan Weed, and
  Alexander~S Wein.
\newblock Estimation under group actions: recovering orbits from invariants.
\newblock {\em arXiv preprint arXiv:1712.10163}, 2017.

\bibitem{bandeira2014multireference}
Afonso~S Bandeira, Moses Charikar, Amit Singer, and Andy Zhu.
\newblock Multireference alignment using semidefinite programming.
\newblock In {\em Proceedings of the 5th conference on Innovations in
  theoretical computer science}, pages 459--470, 2014.

\bibitem{bandeira2020non}
Afonso~S Bandeira, Yutong Chen, Roy~R Lederman, and Amit Singer.
\newblock Non-unique games over compact groups and orientation estimation in
  cryo-{EM}.
\newblock {\em Inverse Problems}, 36(6):064002, 2020.

\bibitem{bandeira2020optimal}
Afonso~S Bandeira, Jonathan Niles-Weed, and Philippe Rigollet.
\newblock Optimal rates of estimation for multi-reference alignment.
\newblock {\em Mathematical Statistics and Learning}, 2(1):25--75, 2020.

\bibitem{bendory2020single}
Tamir Bendory, Alberto Bartesaghi, and Amit Singer.
\newblock Single-particle cryo-electron microscopy: Mathematical theory,
  computational challenges, and opportunities.
\newblock {\em IEEE signal processing magazine}, 37(2):58--76, 2020.

\bibitem{bendory2018toward}
Tamir Bendory, Nicolas Boumal, William Leeb, Eitan Levin, and Amit Singer.
\newblock Toward single particle reconstruction without particle picking:
  Breaking the detection limit.
\newblock {\em arXiv preprint arXiv:1810.00226}, 2018.

\bibitem{bendory2017bispectrum}
Tamir Bendory, Nicolas Boumal, Chao Ma, Zhizhen Zhao, and Amit Singer.
\newblock Bispectrum inversion with application to multireference alignment.
\newblock {\em IEEE Transactions on Signal Processing}, 66(4):1037--1050, 2017.

\bibitem{bendory2022sample}
Tamir Bendory and Dan Edidin.
\newblock The sample complexity of sparse multi-reference alignment and
  single-particle cryo-electron microscopy.
\newblock {\em arXiv preprint, arXiv:2210.15727}, 2022.

\bibitem{bendory2022dihedral}
Tamir Bendory, Dan Edidin, William Leeb, and Nir Sharon.
\newblock Dihedral multi-reference alignment.
\newblock {\em IEEE Transactions on Information Theory}, 68(5):3489--3499,
  2022.

\bibitem{bhamre2015orthogonal}
Tejal Bhamre, Teng Zhang, and Amit Singer.
\newblock Orthogonal matrix retrieval in cryo-electron microscopy.
\newblock In {\em 2015 IEEE 12th International Symposium on Biomedical Imaging
  (ISBI)}, pages 1048--1052. IEEE, 2015.

\bibitem{blumsmith2023machine}
Ben Blum-Smith and Soledad Villar.
\newblock Machine learning and invariant theory.
\newblock {\em arXiv preprtint arXiv:2209.14991}, 2023.

\bibitem{brocker1995representations}
Theodor Br\"{o}cker and Tammo tom Dieck.
\newblock {\em Representations of compact {L}ie groups}, volume~98 of {\em
  Graduate Texts in Mathematics}.
\newblock Springer-Verlag, New York, 1995.
\newblock Translated from the German manuscript, Corrected reprint of the 1985
  translation.

\bibitem{cohen2016group}
Taco~S. Cohen and Max Welling.
\newblock Group equivariant convolutional networks.
\newblock {\em arXiv preprint arXiv:1602.07576}, 2016.

\bibitem{dym2020universality}
Nadav Dym and Haggai Maron.
\newblock On the universality of rotation equivariant point cloud networks.
\newblock {\em arXiv preprint, arXiv:2010.02449}, 2020.

\bibitem{fulton1991representation}
William Fulton and Joe Harris.
\newblock {\em Representation theory}, volume 129 of {\em Graduate Texts in
  Mathematics}.
\newblock Springer-Verlag, New York, 1991.
\newblock A first course, Readings in Mathematics.

\bibitem{hall2015lie}
Brian Hall.
\newblock {\em Lie groups, {L}ie algebras, and representations}, volume 222 of
  {\em Graduate Texts in Mathematics}.
\newblock Springer, Cham, second edition, 2015.
\newblock An elementary introduction.

\bibitem{hewitt1970abstract}
Edwin Hewitt and Kenneth~A. Ross.
\newblock {\em Abstract harmonic analysis. {V}ol. {II}: {S}tructure and
  analysis for compact groups. {A}nalysis on locally compact {A}belian groups}.
\newblock Die Grundlehren der mathematischen Wissenschaften, Band 152.
  Springer-Verlag, New York-Berlin, 1970.

\bibitem{janco2022accelerated}
Noam Janco and Tamir Bendory.
\newblock An accelerated expectation-maximization algorithm for multi-reference
  alignment.
\newblock {\em IEEE Transactions on Signal Processing}, 70:3237--3248, 2022.

\bibitem{kakarala2012bispectrum}
Ramakrishna Kakarala.
\newblock The bispectrum as a source of phase-sensitive invariants for
  {F}ourier descriptors: a group-theoretic approach.
\newblock {\em J. Math. Imaging Vision}, 44(3):341--353, 2012.

\bibitem{kondor2018clebschgordan}
Risi Kondor, Zhen Lin, and Shubhendu Trivedi.
\newblock Clebsch-gordan nets: a fully fourier space spherical convolutional
  neural network.
\newblock {\em arXiv preprint arXiv:1806.09231}, 2018.

\bibitem{lim2022equivariant}
Lek-Heng Lim and Bradley~J. Nelson.
\newblock What is an equivariant neural network?
\newblock {\em arXiv preprint arXiv:2205.07362}, 2022.

\bibitem{liu2021algorithms}
Allen Liu and Ankur Moitra.
\newblock Algorithms from invariants: Smoothed analysis of orbit recovery over
  $ {SO} (3) $.
\newblock {\em arXiv preprint arXiv:2106.02680}, 2021.

\bibitem{ma2019heterogeneous}
Chao Ma, Tamir Bendory, Nicolas Boumal, Fred Sigworth, and Amit Singer.
\newblock Heterogeneous multireference alignment for images with application to
  {2D} classification in single particle reconstruction.
\newblock {\em IEEE Transactions on Image Processing}, 29:1699--1710, 2019.

\bibitem{perry2019sample}
Amelia Perry, Jonathan Weed, Afonso~S Bandeira, Philippe Rigollet, and Amit
  Singer.
\newblock The sample complexity of multireference alignment.
\newblock {\em SIAM Journal on Mathematics of Data Science}, 1(3):497--517,
  2019.

\bibitem{singer2011angular}
Amit Singer.
\newblock Angular synchronization by eigenvectors and semidefinite programming.
\newblock {\em Applied and Computational Harmonic Analysis}, 30(1):20--36,
  2011.

\bibitem{smach2008generalized}
Fethi Smach, Cedric Lema\^{i}tre, Jean-Paul Gauthier, Johel Miteran, and
  Mohamed Atri.
\newblock Generalized {F}ourier descriptors with applications to objects
  recognition in {SVM} context.
\newblock {\em J. Math. Imaging Vision}, 30(1):43--71, 2008.

\end{thebibliography}

\appendix

\section{Representation theory}
\label{sec:representation_theory}

\subsection{Terminology on representations}
Let $G$ be a group. A  representation of $G$ is a homomorphism, 
$G \stackrel{\pi} \to \GL(V)$, where $V$ is a vector space over a field
and $\GL(V)$ is the group
of invertible linear transformations $V \to V$. Given a representation
of a group $G$, we can define an action of $G$ on $V$ by $g \cdot v = \pi(g)v$.
Since $\pi(g)$ is a linear transformation, the action of $G$ is necessarily linear, meaning that for any vectors $v_1, v_2$ and scalars $\lambda, \mu \in \C$
$g \cdot ( \lambda v_1 + \mu v_2) = \lambda(g \cdot v_1) + \mu (g \cdot v_2)$.
Conversely, given a linear action of $G$ on a vector space $V$, we obtain
a homomorphism $G \to \GL(V)$, $g \mapsto T_g$, where
$T_g \colon V \to V$ is the linear transformation $T_g(v) = (g \cdot v)$.
Thus, giving a representation of $G$ is equivalent to giving a linear action
of  $G$ on a vector space $V$. Given this equivalence, we will follow
standard terminology and refer to a vector space $V$ with a linear action
of $G$ as a {\em representation of $G$}.

A representation $V$ of $G$ is {\em finite dimensional} if $\dim V < \infty$.
In this case, a choice of basis for $V$ identifies $\GL(V) = \GL(N)$,
where $N = \dim V$. If $V$ is a complex vector space with
Hermitian inner product $\langle\cdot,\cdot \rangle$
on $V$,  we say that a representation is {\em unitary} if for any two vectors
$v_1, v_2 \in V$
$\langle g \cdot v_1,   g \cdot v_2 \rangle = \langle v_1,   v_2 \rangle$.
Likewise if $V$ is a real vector space with inner product $\langle \cdot, \cdot \rangle$ then we say that the representation is orthogonal if
the action of $G$ preserves the inner product.
If we choose an orthonormal basis for $V$, then the representation
of $G$ is unitary (resp. orthogonal) if and only if the image of $G$ under the homomorphism
$G \to \GL(N)$ lies in the subgroup $U(N)$ (resp. $O(N)$) of unitary (resp.
orthogonal) matrices.

A representation $V$ of a group $G$ is {\em irreducible} if $V$ contains
no non-zero proper $G$-invariant subspaces.

\subsection{Representations of compact groups}
Any compact group $G$ has a $G$-invariant measure called a Haar measure.
The Haar measure $dg$ is typically normalized so that $\int_G dg =1$.
If $V$ is a finite-dimensional representation of a compact group
and $(\cdot,\cdot)$ is any Hermitian inner product, then
the inner product $\langle \cdot,\cdot\rangle$ defined by the formula
$\langle v_1, v_2 \rangle = \int_G (g\cdot v_1, g \cdot v_2)\;dg$ is $G$-invariant. As a consequence we obtain the following fact.
\begin{proposition}
  Every finite dimensional representation of a compact group is unitary.
\end{proposition}
Using the invariant inner product we can then obtain the following decomposition
theorem for finite dimensional representations of compact group.
\begin{proposition}
  Any finite dimensional representation of a compact group
  decomposes into a direct sum of irreducible representations.
\end{proposition}

If $V$ is a representation, then $V^G = \{v \in V| g \cdot v = v\}$
is a subspace which is called the subspace of invariants.

\subsection{Schur's Lemma}
A key property of irreducible unitary representations is Schur's Lemma.
Recall that  a linear transformation $\Phi $ is $G$-invariant if $g \cdot \Phi v = \Phi g \cdot v$.

\begin{lemma} \label{lem.schur}
  Let $\Phi \colon V_1 \to  V_2$ be a $G$-invariant linear transformation
  of finite dimensional irreducible representations
  of a group $G$ (not necessarily compact).
  Then, $\Phi$ is either zero or an isomorphism.
  Moreover, if $V$ is a finite dimensional irreducible unitary representation
  of a group $G$ then any $G$-invariant linear transformation
  $\phi \colon V \to V$ is multiplication by a scalar.
\end{lemma}

\subsection{ Dual, $\Hom$ and tensor products of representations} \label{sec.adjoint}
If $V_1$ and $V_2$ are representations of a group $G$, 
then the vector space $\Hom(V_1, V_2)$ of linear transformations
$V_1 \to V_2$ has a natural linear action of $G$ given by
the formula $(g\cdot A)(v_1) = g \cdot A(g^{-1}v_1)$.
In particular, if $V$ is a representation of $G$, then $V^* = \Hom(V, \C)$
has a natural action of $G$ given by the formula $(g \cdot f)(v) = f(g^{-1}v)$.

A choice of inner product on $V$ determines an identification of vector spaces
$V =V^*$, given by the formula $v \mapsto \langle \cdot , v \rangle$.
If $V$ is a
unitary representation of $G$ then with the identification
of $V=V^*$ the dual action of $G$ on $V$ is given by the formula
$g \cdot_* v = \overline{g} \cdot v$.
Likewise, if $V_1$ and $V_2$ are two representations
then we can define an action of $G$ on $V_1 \otimes V_2$ by the formula
$g\cdot (v_1 \otimes v_2) = (g \cdot v_1) \otimes (g \cdot v_2)$.

Given two representations spaces $V_1, V_2$ there is an isomorphism of representations
$V_1 \otimes V_2^* \to \Hom(V_2, V_1)$ given by 
the formula $v_1 \otimes f_2 \mapsto \phi$, where the linear transform
$\phi \colon V_2 \to V_1$
is defined by the formula
$\phi(v_2) = f_2(v_2) v_1$.
In particular, we can identify $V \otimes V^*$ with $\Hom(V,V)$.
\section{Compact Lie groups} \label{sec.liegroup}
A {\em compact Lie group} is a compact differentiable manifold which is also a group
and with the property that the multiplication and inverse maps are differentiable. A compact Lie group is a {\em torus} if it is isomorphic
to $(S^1)^n$ for some $n$. A fundamental result in the theory of Lie groups
is that every maximal torus in a compact Lie group has the same dimension.
The {\em rank} of a compact Lie group is the dimension of any maximal
torus. For example the rank of $\SU(n)$ is $n-1$ since the set of
determinant one diagonal matrices with non-zero entries of the form $e^{\iota \theta} \in S^1$ is a maximal torus.

The {\em Lie algebra} of a Lie group $G$ is the tangent space
to the Lie algebra at the identity element. 
A connected compact Lie group is {\em simple} if has no non-trivial connected normal subgroups. If $G$ is a simple Lie group then its Lie algebra $\mathfrak{g}$ is a simple Lie algebra meaning that it has no non-trivial
proper ideals. Given a simple Lie algebra $\mathfrak{g}$ there is a unique
simply connected simple compact Lie group $G$ whose Lie algebra is $\mathfrak{g}$.

In this paper we are concerned with representations of the following
simple groups call the {\em classical Lie groups}.
\begin{enumerate}
\item The group $\SU(n)$ of determinant one $n \times n$ unitary matrices.
$\SU(n)$ has rank $n-1$ and is simply connected.

\item The special orthogonal group $\SO(n)$ 
is the group of determinant-one real $n \times n$ matrices $A$ that
satisfy the condition $AA^t = \Id_n$. This group has
rank $\lfloor n/2 \rfloor$ and is simple if $n \geq 3$. The group
$\SO(2)$ is the circle group and is therefore abelian so it is not simple.
The group $\SO(n)$ is not simply connected. The simply connected group
with the same Lie algebra is the {\em spin group} $\Spin(n)$.

\item The symplectic group $\Sp(2n)$ is the group of $2n \times 2n$ unitary matrices $A$ which satisfy the condition $JA =\overline{A}J$ 
where $J$ is the block diagonal matrix $\begin{pmatrix} 0 & -I\\I & 0\end{pmatrix}.$ The symplectic group is simply connected and has rank $n$.

\end{enumerate}
\end{document}